\documentclass{llncs}

\usepackage{amsmath,amssymb,amsfonts}
\usepackage[dvipdfmx]{graphicx}
\usepackage{listings}
\usepackage{xcolor}
\usepackage{url}
\usepackage[caption=false]{subfig}
\usepackage{multirow}
\usepackage[linesnumbered,ruled,vlined]{algorithm2e}
\usepackage{mathtools}
\usepackage{wrapfig}
\usepackage{booktabs}
\usepackage{here}
\usepackage{mdframed}

\usepackage{article}

\mdfsetup{ %
    skipabove=1em, skipbelow=.5em, %
    innertopmargin=.5em, innerbottommargin=.5em, %
    innerleftmargin=.5em, innerrightmargin=.5em, %
    }

\mdfdefinestyle{ranswer}{ %
    backgroundcolor=brown!10, %
    linecolor=black } 

\begin{document}

\mainmatter

\title{SMT-Based Model Checking of \\ Industrial Simulink Models}

\author{
    Daisuke Ishii\inst{1} \and 
    Takashi Tomita\inst{1} \and 
    Toshiaki Aoki\inst{1} \and \\
    The Quyen Ngo\inst{2} \and 
    Thi Bich Ngoc Do\inst{3} \and 
    Hideaki Takai\inst{4} }
\institute{
    Japan Advanced Institute of Science and Technology, Japan
    \and VNU University of Science, Vietnam
    \and Posts and Telecommunications Institute of Technology, Vietnam 
    \and GAIO Technology Co., Japan 
}

\maketitle

\begin{abstract}
    The development of embedded systems requires formal analysis of models such as those described with MATLAB/Simulink.
    However, the increasing complexity of industrial models makes analysis difficult.
    This paper proposes a model checking method for Simulink models using SMT solvers.
    The proposed method aims at
    (1) automated, efficient and comprehensible verification of complex models,
    (2) numerically accurate analysis of models, and
    (3) demonstrating the analysis of Simulink models using an SMT solver (we use Z3).
    It first encodes a target model into a predicate logic formula in the domain of mathematical arithmetic and bit vectors.
    We explore how to encode various Simulink blocks exactly.
    Then, the method verifies a given invariance property using the $k$-induction-based algorithm that extracts a subsystem involving the target block and unrolls the execution paths incrementally.
    In the experiment, we applied the proposed method and other tools to a set of models and properties.
    Our method successfully verified most of the properties including those unverified with other tools.
\end{abstract}

\section{Introduction}

Complex embedded systems are developed using a model-based approach,
in which a \emph{model} of a system is developed virtually before the actual implementation~\cite{LeeSeshia2017};
typical development targets are vehicles and robots.
\emph{MATLAB/Simulink} (Sect.~\ref{s:simulink}) is a tool for developing cyber-physical system (CPS) models.
It provides a graphical language and a numerical simulation engine. 

As ISO-26262 recommends,
formal analysis of models is important to assure the quality of products 
in the model-based development.
A MATLAB toolbox \emph{Simulink Design Verifier (SLDV)}~(Sect.~\ref{s:xp}) provides a set of blocks to represent properties and dedicated model checking functions.
Notably, checking invariance properties plays a crucial role in test generation.
%
%
However, as industrial models become complex, several issues arise in the formal analysis:
\begin{itemize}
\item \emph{Scalability issue} due to the increase in the time taken by checking properties.
\item \emph{Reliability issue} due to the approximation applied during model checking by the tools such as SLDV.
\item \emph{Explainability issue}.
    The detail of the model checking process and the underlying ``formal methods'' of the tool are unknown.
\end{itemize}

\emph{SMT solvers}~\cite{Kroening2016a} are a core technology of formal methods~\cite{Biere2018}.
They handle decision problem instances in various domains e.g. reals, integers, and bit vectors.
Recently, it has become possible to verify properties on floating-point (FP) numbers~\cite{Brain2019}.
In terms of application, they have been applied to analysis of Simulink models (e.g. \cite{Ren2016,Filipovikj2019,Bourbouh2020}).
Yet state-of-the-art solvers are efficient, their scalability is limited in principle; in our preliminary experiments, the execution time blew up when analyzing industrial Simulink models directly.

The objective of this paper is to realize an invariance model checking method that is efficient, formal and comprehensible.
We also aim at a feasibility study of analyzing industrial Simulink models using SMT solvers.
The contributions of this paper are summarized as follows:
\begin{itemize}
    \item \emph{SMT-based model checking method}. 
        We consider invariance properties and verify them with SMT-based model checking (Sect.~\ref{s:method}). 
        The method consists of an encoder from Simulink models to logic formulas (in SMT-LIB format) and a model checker.
        We present two encoding methods (Sect.~\ref{s:enc}):
        \emph{Approximate encoding} based on mathematical numerals and \emph{exact encoding} based on bit vectors.
        For model checking, we propose an algorithm that iteratively applies $k$-induction while 
        expanding the paths unrolled and the local subsystem scope (Sect.~\ref{s:mc}).
        %
    \item \emph{Experimental results with artificial and industrial Simulink models}.
        Sect.~\ref{s:xp} reports the results of the verification of nine models;
        we experimented using the proposed method, SLDV and CoCoSim for comparison.
        We explain that our method processed the models correctly and effectively;
        the advantages over the other tools and the validity of using our method in industrial settings are discussed.
\end{itemize}

\section{Simulink}
\label{s:simulink}

\emph{Simulink}\footnote{\url{https://www.mathworks.com/products/simulink.html}.} is a MATLAB toolbox for modeling synchronous and hybrid systems based on a graphical language.
The targets are described as timed models, either with a continuous timeline or with a timeline discretized with a fixed sample time; in this work, we assume the latter.
Simulink models are diagrams representing hierarchical directed graphs with edges, called \emph{lines}, and nodes, called \emph{blocks}, of various kinds.

Example Simulink models are shown in Fig.~\ref{f:ex}.
Model $S_1$ describes an integration with a feedback loop, in which the input value is added with the output value of the previous step with gain $0.9$; 
here, blocks of type \verb|Add|, \verb|Constant|, \verb|Gain| and \verb|Unit Delay| are utilized;
also, block \verb|Saturate| is used to limit the input range to $[-1,1]$.
Each block is configured with its \emph{parameters} such as gain factor, saturation threshold, and data type of the signal to be processed.
Model $S_2$ embeds $S_1$ as a \emph{subsystem} to model in a hierarchical way.
$S_2$ describes a branch using a \verb|Switch| block that outputs a constant signal with a value of $1$ or $2$, depending on whether the input is greater than $5$ or not.
Model $S_3$ is a more complex example with rate transition and a subsystem with an \verb|Enable| port.
Initially, the subsystem \verb|One| is deactivated (outputs 0); it will be activated when \verb|Compare2| outputs $\True$, but the activation occurs at a 10-fold period.
Model $S_4$ exemplifies matrix and bus signals. It outputs a signal that combines two named elements, a scalar signal \verb|e1| and a matrix signal \verb|e2|.

\begin{figure}[!t]
    \centering
    \subfloat[\label{f:ex:1} Model $S_1$: An integrator.]{%
        \includegraphics[width=.435\textwidth]{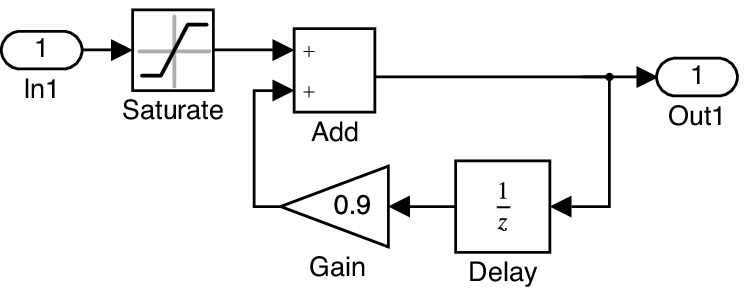} }
    \quad
    \subfloat[\label{f:ex:2} Model $S_2$: A switch after $S_1$.]{%
        \includegraphics[width=.465\textwidth]{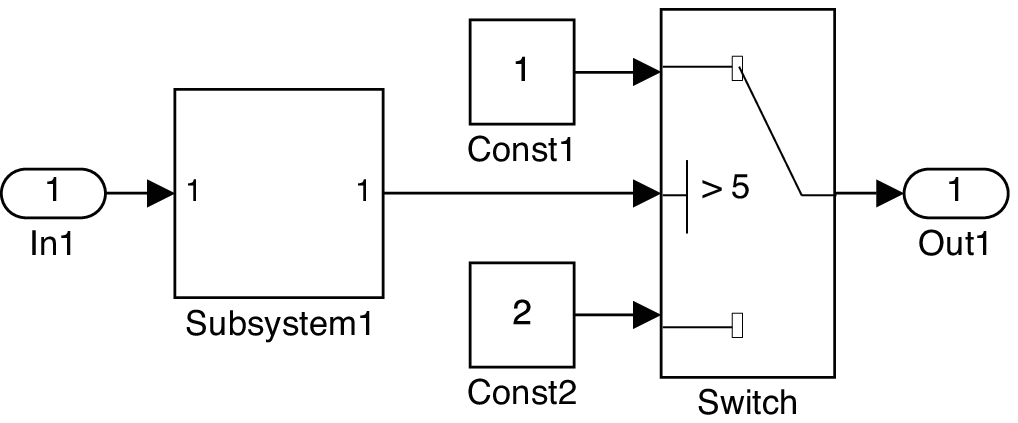} }
    
    \subfloat[\label{f:ex:3} Model $S_3$: An example with multiple rates and an enabled subsystem.]{%
        \includegraphics[height=.185\textwidth]{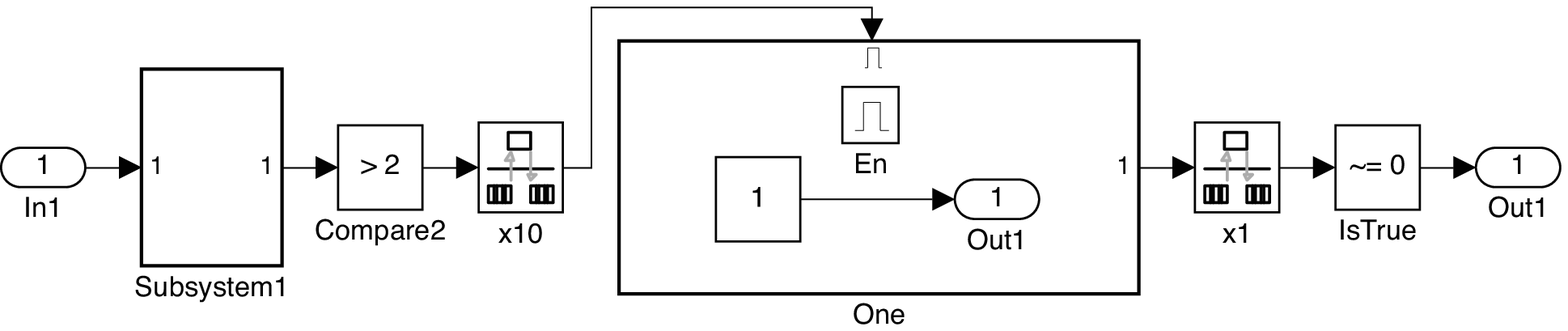} }

    \subfloat[\label{f:ex:4} Model $S_4$: Bus signal.]{%
        \includegraphics[width=.39\textwidth]{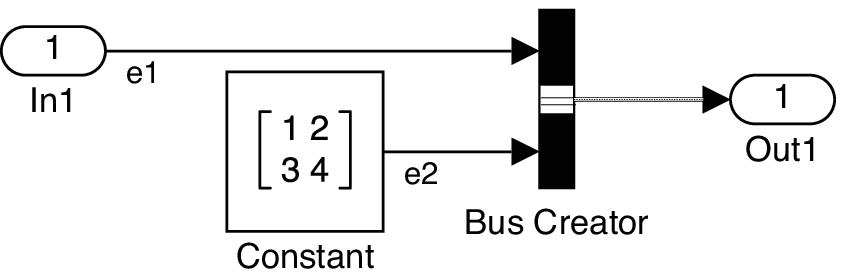} }
    \quad
    \subfloat[\label{f:exec} Example of executing $S_1$.]{%
        \includegraphics[width=.49\textwidth]{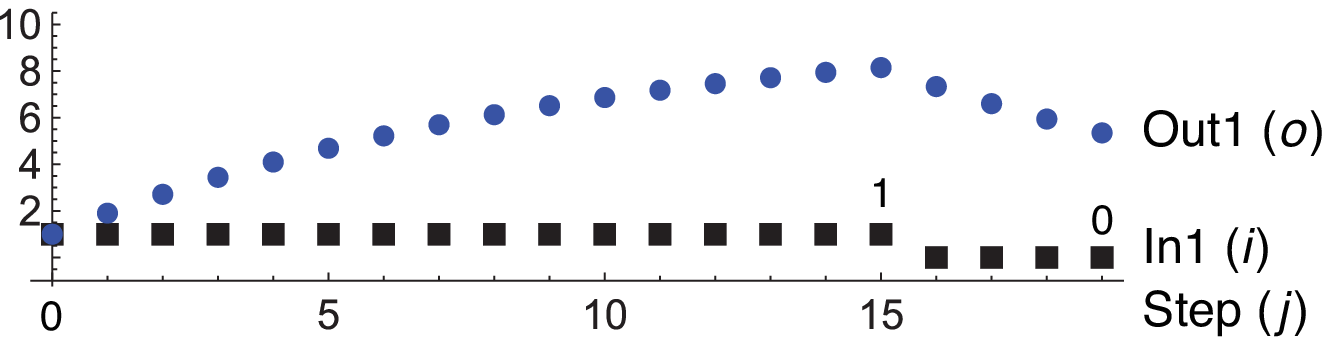} }
    \caption{Example Simulink models.}
    \label{f:ex}
\end{figure}

Primary function of Simulink is numerical simulation, i.e. to obtain output signals of models.
In this paper, we regard a \emph{signal} as a bounded sequence of output values; $j$-th value is output at time $j \times \mathit{st}$ ($j \geq 0$ and $\mathit{st}$ is a configured sample time).
Example input and output signals are shown in Fig.~\ref{f:exec}.

Simulink models have a tree structure consisting of subsystems.
Accordingly, each block in models can be located by a path i.e. a sequence of subsystem names ending with a block name.
We will use this structure to analyze models locally.

\vspace*{-.75em}

\subsubsection{Formalization of Simulink models.}

A discrete-time Simulink model can be regarded as a transition system.\footnote{This formalization may not be applicable to some discrete-time Simulink models, e.g. signal delays for variable lengths. Support for such models is a future work.}

\begin{definition}
    Assume sets $\SIn$, $\SOut$ and $\SSt$ of \emph{input}, \emph{output}, and \emph{state variables}.
    For a set of variables $\mathcal{V}=\{v_1,\ldots,v_n\}$, we denote their \emph{domain} $D_1\!\times\!\cdots\!\times\!D_n$ by $D(\mathcal{V})$.
    A \emph{transition system} $(\Init,\Trans)$ consists of an \emph{initial condition} $\Init \subseteq D(\SSt)$ and a \emph{transition relation} $\Trans \subseteq D(\SSt)\!\times\!D(\SSt)\!\times\!D(\SIn)\!\times\!D(\SOut)$.
\end{definition}

The model $S_1$ is interpreted as $(\Init_1,\Trans_1)$, where:
\begin{align*}
    \Init_1(s) :\Leftrightarrow s = 0, \quad
    \Trans_1(s, s', i, o) :\Leftrightarrow
    o = \max \{-1, \min \{i, 1\}\} + 0.9 s \land
    s' = o.
\end{align*}
Input/output variable $i$/$o$ represents the value of an input/output signal at a step.
The state variable $s$/$s'$ is necessary for \verb|Delay| to represent signal values before/after the transition.
In the same way, $S_2$ is interpreted as $(\Init_2,\Trans_2)$:
\begin{align*}
    \Init_2(s) ~:\Leftrightarrow~& ~ \Init_1(s), \\[-1em]
    \Trans_2(s, s', i, o) ~:\Leftrightarrow~& ~
    \exists i', \exists o', ~
    \Trans_1(s, s', i', o') ~\land~ 
    i = i' ~\land~ o = \begin{cases}
        1 & \text{if $o' > 5$}, \\
        2 & \text{else}.
    \end{cases}
\end{align*}
Predicates are defined based on $(\Init_1,\Trans_1)$. 
Note how variables are handled when subsystemizing; state variables of $S_1$ are inherited to $S_2$; placeholders for the input and output of $S_1$ are prepared locally in the rhs.

We assume that variables are typed as an instance of $\TC$, defined inductively as follows:
\begin{align*}
    \TC &::= \TN ~|~ d \TN ~|~ (d_1 \times\cdots\times d_m) \TN ~|~ \mathit{Bus} \\
    \TN &::= \mathtt{boolean} ~|~ \mathtt{uint}n ~|~ \mathtt{int}n ~|~ \mathtt{double} ~|~ \mathtt{single} ~|~ \mathtt{half} 
\end{align*}
where $d_\Box \in \mathbb{N}$, $m \geq 2$, and $n \in \{8,16,32,64\}$.
Each term of $\TC$ represents scalars, $d$-ary vectors, (possibly higher-dimension) matrices, and buses. 
Buses are concatenated values that combine named members of certain types; see Sect.~\ref{s:enc:complex} for how buses are analyzed.
$\TN$ consists of Boolean type, unsigned/signed integer types and three types for FP numbers.%
\footnote{Other than these types, there are types for fixed-point numbers, strings, enumeration values, and user-defined $\mathtt{ValueType}$ objects; support of these types is a future work.}
Basically, Simulink models are statically typed based on their descriptions and dialog settings.
For instance, the variables $i$, $o$ and $s$ of $S_1$ can be typed as scalar numerical type $t$, e.g. $\mathtt{uint8}$ and $\mathtt{double}$.
Also, they can be typed as $dt$, where $d \in \mathbb{N}$, by configuring each block for element-wise processing of vector values.

Signals obtained by numerical simulation are formalized by execution paths.
\begin{definition}
    Given $(\Init,\Trans)$,
    \emph{execution paths} (or \emph{executions}) of length $k$ are
    \begin{equation*}
        s_{-1} \xrightarrow{i_0/o_0} s_0 \xrightarrow{i_1/o_1} s_1 \cdots s_{k-2} \xrightarrow{i_{k-1}/o_{k-1}} s_{k-1},
    \end{equation*}
    where $s_\Box \in \Dom{\SSt}$, $i_\Box \in \Dom{\SIn}$, $o_\Box \in \Dom{\SOut}$, $\Init(s_{-1})$, and $\Trans(s_{j-1},\AB s_j, \AB i_j, \AB o_j)$ holds for $j \in [0,k-1]$.
    \emph{Input, output} and \emph{state signals} are the traces 
    $i_0 \cdots i_{k-1}$,
    $o_0 \cdots o_{k-1}$, and 
    $s_{-1} \cdots s_{k-1}$ of an execution path.
\end{definition}
The input and output signal values at the initial time are represented by $i_0$ and $o_0$.
Signals can be depicted as in Fig~\ref{f:exec}.

\section{SMT-Based Model Checking}
\label{s:method}

\begin{figure}[t]
    \centering
    \includegraphics[width=1.01\textwidth]{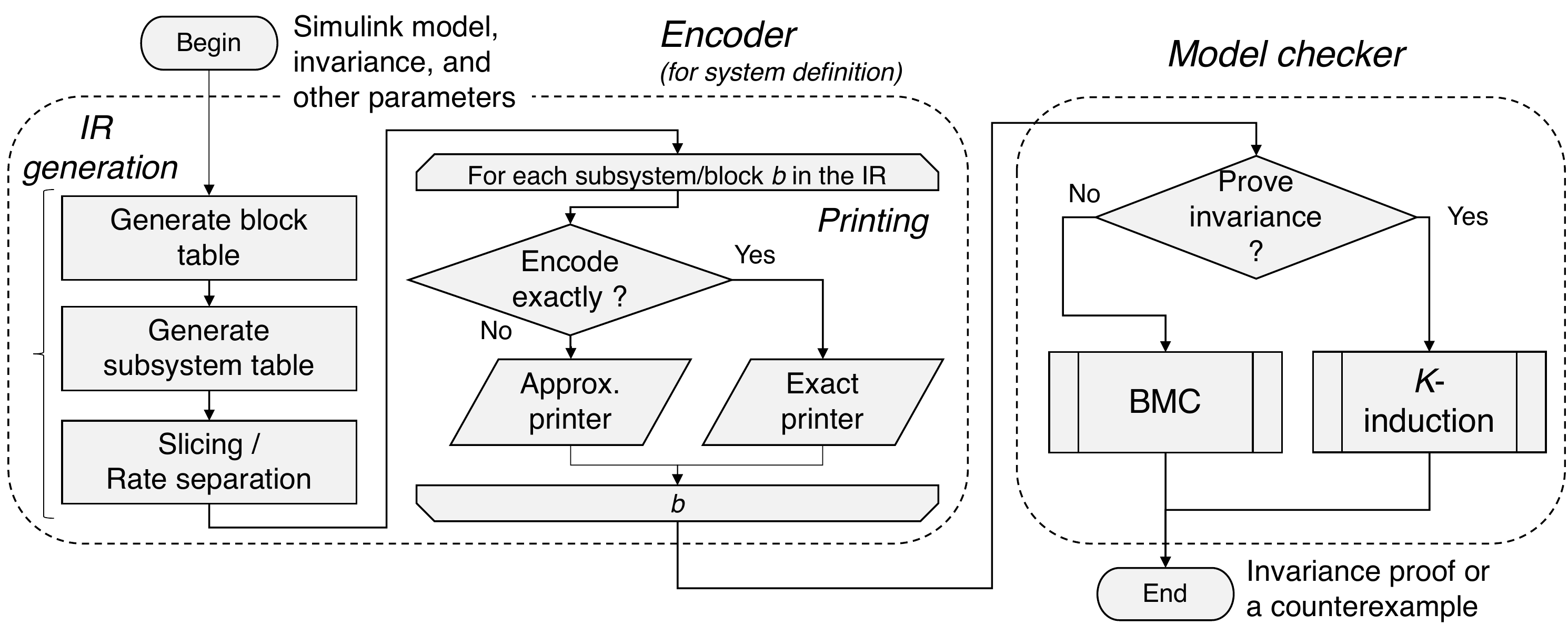}
    \caption{\label{f:method} Process of SMT-based model checking.}
\end{figure}

\emph{SMT (satisfiability modulo theories) solvers} are automated provers for the satisfiability of logic formulas that involve predicates in various theories e.g. integer, real and FP arithmetic.
In this paper, we apply a representative implementation \emph{Z3}\footnote{\url{https://github.com/Z3Prover/z3}.}
to the analysis of Simulink models.

We assume an invariance property of a Simulink model and verify that it holds for the model or violated in an execution.
Typically, such properties can represent a test objective; a counterexample corresponds to a test case and a valid invariance indicates a dead logic.
\begin{definition}
    Given a model $(\Init,\Trans)$, an \emph{invariance} is described by a formula $\Box\phi$, where $\phi \subseteq \Dom{\SIn}\!\times\!\Dom{\SSt}$ is a predicate on the input and state variables.
    Assume an execution path involving input signal $i_0\cdots i_{k-1}$ and state signals $s_{-1}\cdots s_{k-1}$. Then, it is a \emph{counterexample} if $\neg\phi(s_{j-1},i_j)$ holds for a $j\in[0,k\!-\!1]$.
    Invariance $\Box\phi$ \emph{holds} for a model iff there is no counterexample.
\end{definition}
For example, control condition of the \verb|Switch| block in $S_2$ is described by an invariance $\Box (o' > 5)$, where $o'$ represents $\max\{-1,\min\{i,1\}\} + 0.9 s$.
$S_2$ does not satisfy $\Box(o' > 5)$ as shown in Fig.~\ref{f:exec}.
The objectives in Simulink coverage testing are regarded as invariance properties defined for some block types.%
\footnote{\url{https://www.mathworks.com/help/slcoverage/ug/model-objects-that-receive}\AB\texttt{-coverage.html}.}

In this paper, we propose a process illustrated in Fig.~\ref{f:method} for checking an invariance $\Box\phi$ of a Simulink model $(\Init,\Trans)$. It mainly consists of two parts:
\begin{enumerate}
\item Encode a target Simulink model and a property into the input format of SMT solvers (Sect.~\ref{s:enc}).
    It generates a definition of the transition system $(\Init,\Trans)$ represented by the model.
    The predicate $\phi$ is instrumented in the definition.
    The process consists of generation of intermediate representation (IR) of the model and printing of the IR.
    The process involves several steps to handle industrial models.
        Notably, there are 
        a printing step that can generates exact machine representation of numerals (Sect.~\ref{s:enc:exact}), and
        steps for slicing and separation of different rate portions of a model (Sect.~\ref{s:enc:complex}).
\item Model checking of the invariance $\Box\phi$ (Sect.~\ref{s:mc}).
    Verification is basically done by encoding a bounded execution path of the model and by feeding it to an SMT solver.
    We use two methods: A bounded model checking (BMC) method for falsification and a $k$-induction method to prove the invariance.
    We propose an iterative process regarding the parameter $k$ (Alg.~\ref{a:kind}) and a strategy to efficiently expand the target subsystem.
\end{enumerate}

{The basic techniques employed by the proposed method are known ones;
the results of encoding are similar to those obtained by a combination of the CoCoSim~\cite{Bourbouh2020} and Kind2~\cite{Champion2016b} tools;
BMC and $k$-induction are basic SMT-based methods.
In this work, we extend the techniques; for example, we support the encoding of industrial models and we examine a subsystem-wise model checking for efficiency;
see also discussions in Sect.~\ref{s:related}.}

%


When a subsystem $S'$ of a model $S$ is given a property $\Box\phi$, it can be verified \emph{locally} for $S'$ or \emph{globally} for $S$.
Our subsystem-wise method starts from a local verification and gradually attempts the verification for the parent systems until the invariance is proved.
Let $\mathcal{V}_\Box$ and $\mathcal{V}_\Box'$ be the sets of variables for $S$ and $S'$.
To verify globally, we consider a translation $\phi\!\uparrow\!S$ of the predicate $\phi$ in $S$, which is straightforward since $\SSt' \subseteq \SSt$ and the Simulink model $S$ should describe a relation in $D(\SIn)\!\times\!D(\SSt)\!\times\!D(\SIn')$ (note that $\phi \subseteq D(\SIn')\!\times\!D(\SSt')$).
There is the following relationship between the two:
\begin{proposition} \label{prop:local}
    If $\Box\phi$ holds for $S'$, $\Box(\phi\!\uparrow\!S)$ holds for $S$.
\end{proposition}
\begin{proof}
The contraposition obviously holds because a counterexample of $\Box\phi$ can be extracted from that of $\Box(\phi\!\uparrow\!S)$.
\end{proof}

\section{SMT-LIB Encoding of Simulink Models}
\label{s:enc}

\begin{figure}[t!]
    \lstset{frame=single}
    \lstset{numbers=left}
%
    \begin{lstlisting}[basicstyle=\ttfamily\footnotesize, columns=flexible, keepspaces=true]
;; Variable representing the current step.
(declare-const curr_step Int)
;; Whether to verify the induction step.
(declare-const flag_kind Bool)

;; Specification of subsystem S1.
(define-fun init1 ((s@0 Real)) Bool (= s@0 0))
(define-fun trans1 
  ((c Int) (s@0 Real) (s@1 Real) (i Real) (o Real)) Bool
  (let ((lv (saturate 1 (- 1) i)))
    (and (= o (+ lv (* 0.9 s@0))) (= s@1 o)) ))

;; Specification of the parent system S2.
(define-fun init2 ((s@0 Real)) Bool (init1 s@0))
(define-fun trans2
  ((c Int) (s@0 Real) (s@1 Real) (i Real) (o Int)) Bool
  (exists ((i_ Real) (o_ Real))
    (and (trans1 c s@0 s@1 i_ o_) (= i_ i) (= o (ite (> o_ 5) 1 2))

      ;; Objective instrumentation.
      (=> (= c curr_step) (not (> o_ 5)))
      ;; Assumption for the induction step.
      (=> flag_kind (=> (< c curr_step) (> o_ 5))) )))

;; Encoding of the execution path.
(declare-const s@i Real) (assert (init2 s@i))

(declare-const s@0 Real) (declare-const i@0 Real) 
(declare-const o@0 Int) (assert (trans2 0 s@i s@0 i@0 o@0))

(declare-const s@1 Real) (declare-const i@1 Real) 
(declare-const o@1 Int) (assert (trans2 1 s@0 s@1 i@1 o@1))

;; Check the reachability at step 1.
(check-sat-assuming (and (= curr_step 1) (not flag_kind)))
\end{lstlisting}
    \caption{\label{f:s1:enc}Simulink models $S_1$ and $S_2$ encoded in SMT-LIB.}
\end{figure}

\emph{SMT-LIB}~\cite{Barrett2021}
is an input format for SMT solvers, which has a LISP-like prefix grammar.
Here, we describe the encoding method with an example.
Fig.~\ref{f:s1:enc} lists an encoded SMT-LIB description of the model $S_2$ (and $S_1$ as a subsystem).
At \textbf{Lines~6--11}, predicates $\Init_1$ and $\Trans_1$ of $S_1$ are defined as \lstinline|Bool|-valued functions \lstinline|init1| and \lstinline|trans1|.
\textbf{Lines~13--18} define $(\Init_2,\Trans_2)$ in the same way.
The encoding process is either \emph{approximate} or \emph{exact} (Sect.~\ref{s:enc:exact}); the example is an approximate encoding.
Assuming that variables $s$, $s'$, $i$ and $o$ of $S_1$ are typed as FP numbers, 
they are encoded with variables \lstinline|s@0|, \lstinline|s@1|, \lstinline|i| and \lstinline|o| of sort \lstinline|Real|, which represents mathematical rational numbers.
In the same way, variable $o$ of $S_2$ is sorted as unbounded integers.
%
\textbf{Lines~20--23} instrument the invariance $\Box(\texttt{o\char`_}>5)$.
Two global variables \lstinline|curr_step| and \lstinline|flag_kind| and the argument \lstinline|c| of transition predicates are introduced for the verification (Sect.~\ref{s:mc}).
\textbf{Lines~25--32} describe a length-2 execution path where values at each step are parameterized by fresh variables.
The predicate \lstinline|init2| is asserted for the initial step $-1$ and \lstinline|trans2| is asserted for the later steps.
State variables are shared between two steps.
Finally, at \textbf{Line~35}, command \lstinline[mathescape]|(check-sat-assuming $c$)| will invoke a solving process;
it will assume the argument constraint $c$ temporarily.
With setting a step number to \lstinline|curr_step| and disabling \lstinline|flag_kind|, the violation of the invariance at step~1, i.e. reachability to the state such that $\neg(\texttt{o\char`_}>5)$, should be checked.
This example will result in $\Unsat$; unrolling the execution path up to step~6 will result in $\Sat$.


%

\subsection{Exact Encoding of Machine-Representable Numbers}
\label{s:enc:exact}

For reliable analysis, we propose to encode FP numbers and integers using the vocabulary provided by the SMT-LIB theories \lstinline|FloatingPoint| and \lstinline|BitVector|.
This encoding method exactly describes rounded values, overflow cases, etc.
Solving exact formulas tends to be expensive; thus, we use this method along with the approximate method.
\lstinline|FloatingPoint| 
provides the sorts (e.g. \lstinline|Float64| for double-precision FP numbers), arithmetic operators (e.g. \lstinline|fp.add| for addition), and utility functions (e.g. \lstinline[mathescape]|(_ fp.to_sbv $n$)| that converts an FP number to a signed bit vectors of length $n$).
For machine integers, we prepare necessary vocabularies based on \lstinline|BitVector| as in \cite{Baranowski2020};
we use the sort \lstinline[mathescape]|(_ BitVec $n$)| to represent $n$-bit integers and define the functions for arithmetic operations, e.g. \lstinline|int64.add| (they are  defined in the beginning of the encoder outputs).
%

For example, the first equation in Line~11 of Fig.~\ref{f:s1:enc} is encoded as:
\begin{lstlisting}[basicstyle=\ttfamily\small, xleftmargin=0em]
(= o (fp.add RNE lv (fp.mul RNE (fp #b0 #b01111111110 
  #b1100110011001100110011001100110011001100110011001101 ) s@0 )))
\end{lstlisting}
Variables \lstinline|o|, \lstinline|lv| and \lstinline|s@0| are of sort \lstinline|Float64|.
\lstinline|RNE| represents a rounding mode. 

\subsection{Encoding of Complex Simulink Models}
\label{s:enc:complex}

This section describes techniques for more complex models.

A subsystem $(\Init,\Trans)$ can be executed conditionally by adding an activation port,
e.g. \verb|EnablePort| in Fig.~\ref{f:ex:3}. 
When deactivated, the subsystem outputs the initial value or the previous value.
We encode such subsystems by introducing wrappers for $\Init$ and $\Trans$.
For example, $\Trans$ of the subsystem \verb|One| in Fig.~\ref{f:ex:3} is encoded into the following wrapper predicate:
\begin{lstlisting}[basicstyle=\ttfamily\small, xleftmargin=0em]
(define-fun trans ( (ien Bool) (so@0 %$t$%) (so@1 %$t$%) (o %$t$%) ) Bool
  (ite ien
    ;; Activate the body transition predicate.
    (and (trans_body o) (= so@1 o) )
    ;; Else, output the prev value and keep the state unchanged.
    (and (= o so@0) (= so@1 so@0)) ))
\end{lstlisting}
%
It assumes that the body content (to output constant 1) \lstinline|trans_body| is pre-defined.
Output variable \lstinline|o| of sort $t$ is inherited from \lstinline|trans_body|.
Variables \lstinline|ien|, \lstinline|so@0| and \lstinline|so@1| are introduced to represent a signal input to the \verb|EnablePort| and state variables to keep track of a previous output value when deactivated.
$\Init$ is also wrapped and the initial output value is configured.

Activation using \verb|TriggerPort| is encoded with a wrapper and a behavioral description of the trigger signal such as \verb|Rising| (we pre-define pattern functions as in \cite{Tripakis2005}).
%
We also use wrapper predicates to encode multi-rate models (e.g. $S_3$ in Fig.~\ref{f:ex:3}).
For a subsystem configured to be executed with a slower rate, we encode with a wrapper equipped with a local counter variable,
which computes the activation period and activates the body predicate accordingly.
Multiple rates may be present in a subsystem; in such a case, we run a preprocess to divide the subsystem into separate \emph{dummy} subsystems for each rate (Sect.~\ref{s:enc:impl}).

Encoding of vector signals is simply done by preparing scalar variables for each element of vector values.
For bus signals, we utilize their schema externally specified by a \emph{bus object data type}.%
\footnote{We consider only \emph{nonvirtual} buses.}
Based on the specification, we introduce an SMT-LIB sort that represents bus signal values.
For the model $S_4$ in Fig.~\ref{f:ex:4}, the sort and accessors for the elements are declared as follows:
\begin{lstlisting}[basicstyle=\ttfamily\small, xleftmargin=0em]
(declare-sort BO 0)
(declare-fun BO_e1 (BO) Int) ;; Accessor for member e1.
(declare-fun BO_e2_1_1 (BO) Real) ;;Accessor for member e2.
;; BO_e2_1_2, BO_e2_2_1, and BO_e2_2_2 are also declared.
\end{lstlisting}
Given a value \lstinline|bo| of sort \lstinline|BO|, \lstinline|(BO_e2_1_1 bo)| represents the element at $(1,1)$ of the matrix-typed member \verb|e2|.
Using the sort \lstinline|BO|, $\Trans_4$ is defined as follows:
\begin{lstlisting}[basicstyle=\ttfamily\small, xleftmargin=0em]
(define-fun trans ((c Int) (i Int) (o BO)) Bool
  (and (= (BO_e1 o) i) (= (BO_e2_1_1 o) 1) (= (BO_e2_1_2 o) 2) 
%\hspace{11.5em}%(= (BO_e2_2_1 o) 3) (= (BO_e2_2_2 o) 4) ))
\end{lstlisting}

\subsection{Implementation of the Encoder}
\label{s:enc:impl}

We have implemented the encoder as a MATLAB script (about 9000~LOCs).
The script implements the IR generation process and printers as shown in Fig.~\ref{f:method}.

In IR generation, we first prepare a \emph{block table} (BT), an array of \verb|struct| (record) data where each element represents a block (or a subsystem). 
Most content of a Simulink model
can be accessed via command \verb|get| (or \verb|get_param|) in a script.
So, BT collects necessary information,
e.g. block-line graph structure and type signature of each inport/outport.
Second, we generate another \verb|struct| array, \emph{subsystem table} (ST), that represents the tree structure of subsystems.
Each element corresponds to a transition system and contains lists of variables and a list of child subsystems.
Third, we modify BT and ST to support multi-rate models and to slice the target portion of the content.
Slicing is done by a backward reachability analysis on BT, starting from the objective block of the invariance property.
For each multi-rate ST element, we classify the blocks rate-wise, and then introduce dummy subsystems with dummy inports and outports.

Approximate and exact printers basically translate the content of a ST into an SMT-LIB description.
Each ST element is printed as definitions of corresponding $\Init$ and $\Trans$.
The body of each definition mainly contains assignments to the variables and their rhs are printed by traversing BT.

Our implementation supports encoding of 37 block types (but not all parameter settings).
Unsupported block instances are stubbed with local unconstrained variables in a predicate definition.


%

\section{Model Checking Methods}
\label{s:mc}

As described in Sect.~\ref{s:method}, 
we consider two methods for verifying an invariance $\Box\phi$ of a Simulink model $(\Init,\Trans)$, i.e. BMC and $k$-induction.
In the following, we abbreviate \lstinline[mathescape]|(check-sat-assuming $c$)| to $\mathsf{CSA}(c)$,
\lstinline|curr_step| in Fig.~\ref{f:s1:enc} to $\mathtt{c}$ and \lstinline|flag_kind| to $\mathtt{f}$, respectively.

Given $k \geq 1$, the BMC method searches for a counterexample of length-$k$ or less.
It performs the following steps for each $j \in [0,k\!-\!1]$.
(0)~Assume that $(\Init,\Trans)$ instrumented with $\phi$ has been encoded.
(1)~Encode the length-$(j\!+\!1)$ execution paths.
(2)~Feed the encoded result to an SMT solver and then solve the command 
    $\mathsf{CSA}(\mathtt{c} = j \land \mathtt{f} = \bot)$.
The process is done efficiently by incrementing from $j=1$,
encoding only $\Trans$ for the $j$-th step in (1).

The $k$-induction method~\cite{Sheeran2000} consists of the proofs of the following facts:
\begin{itemize}
    \item \emph{Base case}: ``$\phi$ is invariant for execution paths of length $k\!-\!1$ or less.''
    \item \emph{Induction step}: ``Assume execution paths of length $k$ that are not initialized. Let $j \in [0,k\!-\!2]$, and if $\phi$ holds for every $j$-th step of a path, then $\phi$ holds at the last $(k\!-\!1)$-th step of the path.''
\end{itemize}

\begin{algorithm}[t]
    \SetAlgoLined
    \SetKwInOut{Input}{Input}\SetKwInOut{Output}{Output}
    \Input{Simulink model $(\Init,\Trans)$, Invariance $\Box\phi$, $k \in \mathbb{N}_{\geq 2}$, Encode exactly? $b$}
    \Output{$\True$, $\False$ or $\Maybe$}
    \BlankLine
    $\mathsf{EncodeSystemAndAssert}_b(\Init,\Trans,\phi)$\;
    \For{$j \in [0,k\!-\!1]$}{
        $\mathsf{EncodeTransAndAssert}_b(\Trans,j)$\;
        \textbf{if}\ {$j > 0 \land \mathsf{CSA}(\mathtt{c}\!=\!j \land \mathtt{f}\!=\!\top) = \Unsat$}\ \textbf{then return}\ $\True$;\ \textbf{end}\\
        \If{$\mathsf{CSA}(\mathsf{EncodeInit}_b(\Init) \land \mathtt{c}\!=\!j \land \mathtt{f}\!=\!\bot) = \Sat$}{
            \textbf{If}\ {$(\Init,\Trans)$ is top-level}\ \textbf{then return}\ $\False$;\ \textbf{else return}\ $\Maybe$;\ \textbf{end}
        }
    }
    \Return{$\Maybe$}\;
    \caption{A $k$-induction procedure.}
    \label{a:kind}
\end{algorithm}

Alg.~\ref{a:kind} is an incremental procedure for proving the two facts.
We assume an SMT solver process running in the background.
The algorithm generates SMT-LIB expressions in several stages and feeds them to the solver at each stage.
The definition of the model is generated at \textbf{Line~1} and predicate $\Trans$ reaching step~$j$ is generated at \textbf{Line~3}.
Then, it verifies the base case (\textbf{Line~5}) and induction step (\textbf{Line~4}; checked one iteration after).
Based on the locality (Prop.~\ref{prop:local}), $(\Init,\Trans)$ can be a subsystem, but it is needed to be the top-level system to falsify the invariance (\textbf{Line~6}).

\subsubsection*{Implementation.}

BMC and $k$-induction methods have been implemented as MATLAB scripts (about 2500 LOCs).
In addition to the encoding process, they generate SMT-LIB expressions for execution paths during a verification process.
For SMT solving, they communicate with an external server that wraps Z3 via a TCP socket.

The $k$-induction script conducts verification subsystem-wise, starting from the objective subsystem of the invariance property.
The script repeatedly invoke Alg.~\ref{a:kind}, controlling the following two factors (initially $d := 1$ and $k := 1$):
\begin{enumerate}
\item \emph{Number $d$ of subsystem hierarchies}. We start applying Alg.~\ref{a:kind} on the subsystem ($d = 1$) that is targeted by the property. When the result is inconclusive, we increase the value of $d$, i.e. to invoke Alg.~\ref{a:kind} in the upper hierarchy.
\item \emph{Bound $k$ on the path lengths}. Once the whole hierarchy has been analyzed, we increase the value $k$ exponentially (and repeat for each subsystem again).
    Alg.~\ref{a:kind} is modified to skip the verification for an initial range of $[1,k\!-\!1]$ that has been processed in the previous round.
\end{enumerate}

\section{Experimental Evaluation}
\label{s:xp}

This section describes an empirical evaluation of the proposed method (abbreviated as ``PM'' in the following).
We have verified artificial and industrial Simulink models using our implementation and exiting tools (SLDV and CoCoSim).
Here, the purpose was to answer:
\textbf{RQ1}: Does PM correctly handle the collected Simulink models?
\textbf{RQ2}: How is the scalability of PM?
\textbf{RQ3}: Is the performance of PM better than other tools?

Experiments were conducted on 64-bit Ubuntu 20.04 virtual machines (with 4~cores and 8GB RAM),
running on a 2.2GHz Intel Xeon E5-2650v4 processor (12~cores) with 128GB RAM.
We used MATLAB~R2022a.
Execution time was limited to 1~hour.

\subsubsection{Evaluated tools.}

We used the following three tools.

(i)
\emph{PM} implemented (cf. Sect.~\ref{s:enc:impl} and \ref{s:mc})
as an add-on to a proprietary tool PROMPT.\footnote{\url{https://www.en.gaio.co.jp/products/prompt-2/}.} 
In the experiment, we applied the BMC method to obtain counterexamples of false properties, and the $k$-induction method to prove true properties.
BMC was performed against the entire system, while $k$-induction (Alg.~\ref{a:kind}) was performed subsystem-wise.
Each verification was performed in two ways using either approximate or exact encoder;
we respectively refer to them by ``PM/A'' and ``PM/E.''

(ii)
\emph{SLDV},\footnote{\url{https://www.mathworks.com/products/simulink-design-verifier.html}.} a tool for analyzing Simulink models based on formal methods.
Also, it seems to apply approximations in the analysis.
Among various functions of SLDV,
we used the ``\emph{property proving (PP)}'' function in this experiment.
If a \verb|Proof Objective| block representing an invariance property is added to a Simulink model,
it performs verification of whether it holds or be falsified by a counterexample.
We enabled the option ``FindViolation'' for false properties and ``Prove'' for true properties.

(iii)
\emph{CoCoSim} (version~1.2), a front-end tool for applying formal tools to Simulink.
Its ``prove properties'' function allows invariance verification using Kind2 (version~1.2), a model checker for Lustre programs.
Process of the function consists of translation from Simulink to Lustre and invocation of Kind2.

%

\subsubsection{Target Simulink models.}

We have prepared 5 artificial and 4 industrial examples to ensure that the different types of models could be handled correctly and effectively. 
We refer to the models as \Ex{1}--\Ex{9}; also, we refer to a model with a parameter $p$ as $\Ex{$i$}_p$.
Model~$\Ex{1}_\mathit{th}$ represents $S_2$ in Sect.~\ref{s:simulink} with switching threshold $\mathit{th}$.
Models~$\Ex{2}_\mathit{th}$ and $\Ex{3}_\mathit{th}$ describe second and fourth-order digital filters with threshold $\mathit{th}$ on the output signals.
Model~$\Ex{4}_\mathit{th}$ describes 32 nested counters that are reset at threshold $\mathit{th}$.\footnote{\url{https://github.com/dsksh/sl-examples}.}
Model~\Ex{5} is the multi-rate model $S_3$ with additional logic blocks.
Models~\Ex{6}--\Ex{8} are taken from the Lockheed Martin challenge~\cite{Elliott2016,Mavridou2020}.
We used the first, fourth and sixth problems.
Model~\Ex{9} describes 
a realistic controller using various types of blocks and externally-controlled subsystems.
The size of each model is shown in Table~\ref{t:stat}.

\subsection{Results}

For each model, we considered 2 or 4 invariance properties such that half of them are false and the rest are true;
false and true instances are denoted by $\Ex{$i$}_{\mathrm{P}j}$ and $\Ex{$i$}_{\mathrm{Q}j}$, respectively ($j\in\{1,2\}$).
The properties prepared can be regarded as test objectives for a block or subsystem, such as conditions on input/output signals and activation conditions.
We verified them using the tools.
The execution time of PM and SLDV is shown in Fig.~\ref{f:chart}.
The right side of Table~\ref{t:stat} shows the parameters of PM used for each property;
the number of sliced blocks that belong to the analyzed subsystems, 
the maximum bound $k$ used in the model checking, and
the number $d$ of analyzed subsystem hierarchies.
Results of the SMT-LIB encoding of the models (other than \Ex{9}) are available at 
\url{https://www.dropbox.com/s/ohwo5sq0vcak566/2022_icfem_smt2.zip}.

\subsection{Discussions}

\noindent
\textbf{RQ1 (correctness).} All the conclusive results of PM were confirmed correct.
For false properties, we manually simulated the obtained counterexamples using Simulink and confirmed that they were indeed so.
For true properties, we confirmed that the results of PM were the same as other tools.
Errors on $\Ex{9}_{\textrm{P2}}$ occurred due to unsupported blocks by the encoder.
Some block types, e.g. lookup tables and nonlinear arithmetic functions, are difficult to encode and/or to analyze. Our tool either causes an error when falsifying, or abstract them with stub variables in the encoding for an invariance proof.

\begin{wrapfigure}[7]{r}{0.3\textwidth}
    \vspace*{-3.5em}
    \includegraphics[width=.3\textwidth]{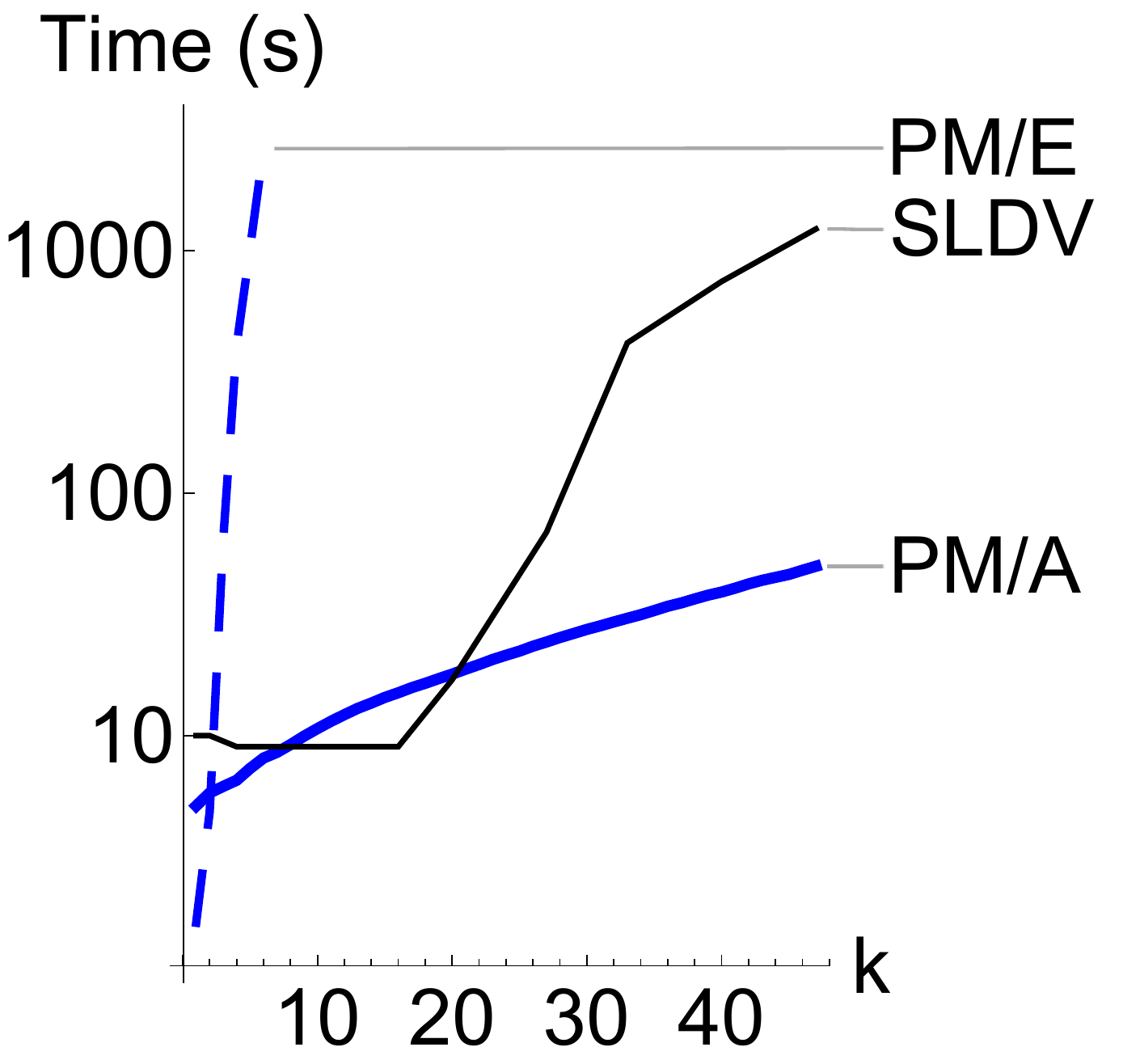}
\end{wrapfigure}

\vspace{.5em}

\noindent
\textbf{RQ2 (scalability).}
Using a parameterized model, we can observe the exponential increase of execution time.
Right figure shows the time needed to falsify the instances $\Ex{1}_{\mathit{th},\mathrm{P1}}$ with several $\mathit{th}$'s.
The scalability of PM/A was enough (better than SLDV on more than half of the instances) to handle most of the prepared models.
PM/E scaled an order of magnitude worse, limiting the number of instances it could handle.
As shown in Table~\ref{t:stat} and Fig.~\ref{f:chart}, the execution time increases basically with the number of encoded blocks, $\text{``\# b's''}\times k$.\footnote{It is likely to depend on other factors, e.g. the form of encoded formulas and the number of solutions; a detailed analysis is omitted from this paper.}
To improve the applicability, slicing and subsystem-wise process of PM is significant. Table~\ref{t:stat} shows that the numbers of encoded blocks were reduced, contributing to the number of instances verified in time.

\begin{table}[!t]
    \centering
    \small
    \caption{Statistics on the models and the solving process. The abbreviations ``s's'' and ``b's'' represent ``subsystems'' and ``sliced blocks,'' respectively. 
        Third to sixth sections correspond to false and true properties.%
        } \label{t:stat}
    \begin{tabular}{l|rr|r@{\quad}r|r@{~~}r|r@{~~}r@{\quad}r|r@{\quad}r@{\quad}r}
        \toprule
        \multicolumn{3}{c}{} 
        & \multicolumn{2}{c}{\raisebox{1em}{\tiny P1}} 
        & \multicolumn{2}{c}{\raisebox{1em}{\tiny P2}}
        & \multicolumn{3}{c}{\raisebox{1em}{\tiny Q1}} 
        & \multicolumn{3}{c}{\raisebox{1em}{\tiny Q2}} \\[-1.25em]
        \cmidrule(lr){4-5}
        \cmidrule(lr){6-7}
        \cmidrule(lr){8-10}
        \cmidrule(lr){11-13}
        Model          & \# blocks & \# s's  &  \# b's & $k$  &             \# b's & $k$ &  \# b's &      $k$ & $d$ & \# b's & $k$ & $d$ \\
        \midrule
        ~$\Ex{1}_\mathit{th}$ &  12 &     1  &      11 &   7  & \multicolumn{2}{c|}{---} &      11 &   1 &  1 & \multicolumn{3}{c}{---} \\
        ~$\Ex{2}_\mathit{th}$ &  16 &     1  &      14 &   2  &                 14 &   7 &      14 &  10 &  2 &     14 &  29 & 2 \\
        ~$\Ex{3}_\mathit{th}$ &  24 &     1  &      22 &  15  & \multicolumn{2}{c|}{---} & 22 & $\geq34$ &  2 & \multicolumn{3}{c}{---} \\
        ~$\Ex{4}_\mathit{th}$ & 290 &    32  &     289 &  12  & \multicolumn{2}{c|}{---} &     289 &   1 &  1 & \multicolumn{3}{c}{---} \\
        ~\Ex{5}               &  33 &     8  &      26 &  21  & \multicolumn{2}{c|}{---} &      25 &   1 &  1 & \multicolumn{3}{c}{---} \\
        ~\Ex{6}               & 479 
                                    &    39  &     321 &   2  &                327 &   1 &     259 &   1 &  1 &    321 &   1 & 2 \\
        ~\Ex{7}               & 291 
                                    &    18  &     116 &   1  &                103 &   1 &      58 &   1 &  2 &      9 &   1 & 3 \\
        ~\Ex{8}               & 712 
                                    &   188  &      27 &   1  &                327 &   1 &      27 &   1 &  3 &    327 &   1 & 3 \\
        ~\Ex{9}               & 574 
                                    &    30  &      19 &   1  &            407 & $\geq1$ &      12 &   1 &  1 &    407 &   1 & 3 \\
        \bottomrule
    \end{tabular}
\end{table}
\begin{figure}[!t]
    \centering
    \subfloat[\label{f:chart:1} Falsification.]{%
        \hspace{-14.9em} \includegraphics[width=\textwidth]{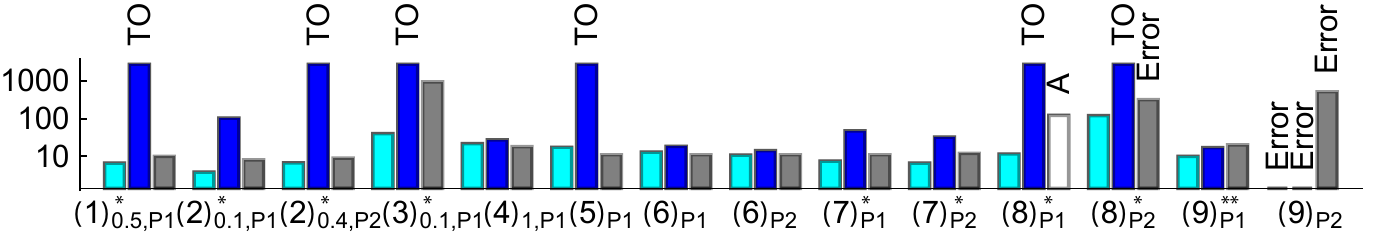}%
        \hspace{-19em} \raisebox{3.5em}{\includegraphics[width=.1\textwidth]{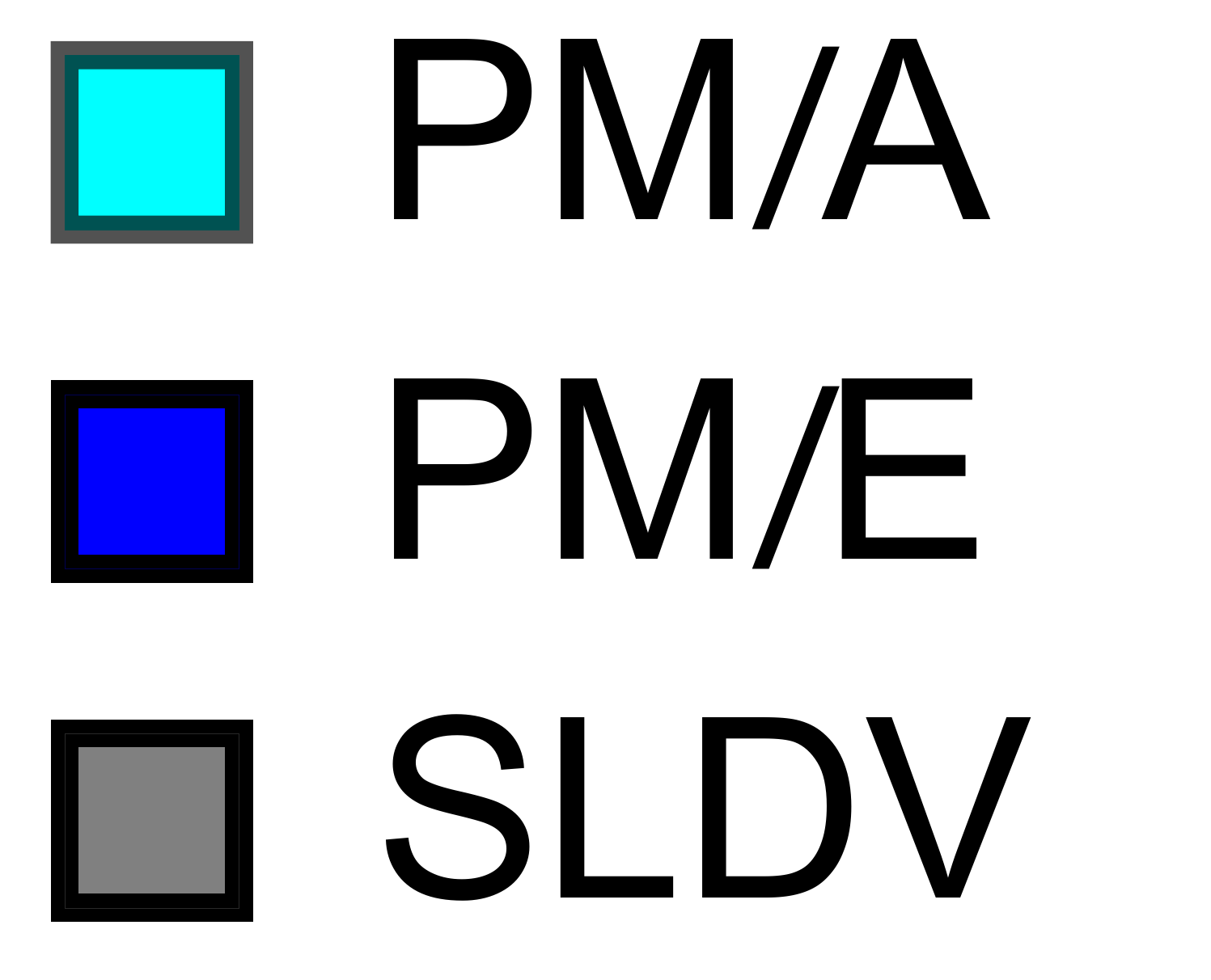}}
        }

    \vspace{-1em}

    \subfloat[Invariance verification.]{%
        \includegraphics[width=\textwidth]{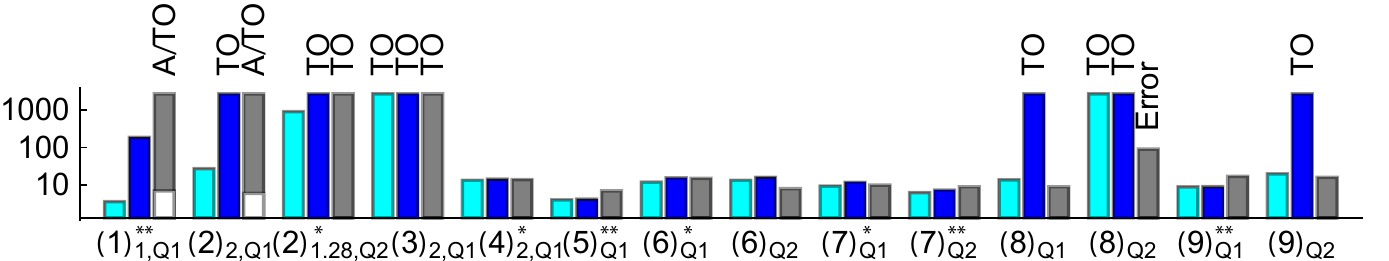} }
    \caption{Execution time in seconds. Results marked with ``TO'' are timeouts. ``A'' means ``solved under approximation (white portion shows time required).'' Superscript * (resp. **) indicates that PM/A (resp. PM/A and PM/E) outperforms SLDV.}
    \label{f:chart}
\end{figure}

\vspace{.5em}

\noindent
\textbf{RQ3 (tool comparison).}
PM/A outperformed the other tools for 17 out of 28 instances.
PM/E was able to handle 15 instances including industrial ones, although its performance was worse than PM/A (it ran out of time for the others).
Notably, we verified 3 instances that could not be handled by other tools ($\Ex{8}_{\mathrm{P2}}$, exact verification of $\Ex{1}_{1,\mathrm{Q1}}$ and approximate verification of $\Ex{2}_{1.28,\mathrm{Q2}}$).
Regarding the encoding method, PM/E should be used to prove a property reliably.
When falsifying properties, counterexamples obtained by PM/A can be certified by a simulation using Simulink.

SLDV resulted in timeouts for 4 true instances.
On $\Ex{1}_{1,\textrm{Q1}}$ and $\Ex{2}_{2,\textrm{Q1}}$, it first indicated ``valid under approximation'' and then ran out of time to prove the validity accurately.
It is not clear whether the ``approximation'' method is similar to PM.
Errors on $\Ex{8}_{\textrm{P2/Q2}}$ and $\Ex{9}_{\textrm{P2}}$ were ``due to nonlinearities.''

CoCoSim handled only model $\Ex{1}_{th}$.
5 models were resulted in errors during conversion to Lustre (due to unsupported blocks).
Verification of 3 models failed in the model checking process of Kind2.


%

\section{Related Work}
\label{s:related}


Model checking techniques using SAT/SMT solvers have been applied to various domains~\cite{Biere2018}.
%
Bourbouh et al.~\cite{Bourbouh2020} have developed the CoCoSim tool for Simulink models using an SMT-based model checker, Kind2~\cite{Champion2016b}.
The basic process in this paper is similar but we provide an exact encoding method and more support for industrial models.
Filipovikj et al.~\cite{Filipovikj2019} have proposed a bounded model checking method for a subset of Simulink models.
A related technology is model checking for the synchronous language Lustre~\cite{Halbwachs1993,Caspi2003}.
Versions of the Kind tool~\cite{Hagen2008a,Kahsai2011,Champion2016b} and Zustre~\cite{Kahsai2018} have been developed.
They have implemented techniques such as IC3~\cite{Bradley2011}, parallel solving, and Horn clause encoding;
they can be implemented in our method in the future.
Kind2 supports accurate encoding with machine integers; our method handles FP numbers in addition.

SMT-based methods require formalization of the target Simulink models.
Some of the existing work has formalized via translation to Lustre~\cite{Caspi2003,Tripakis2005,Bourbouh2020}.
The basic concepts of Simulink are naturally mapped to the Lustre counterparts.
A node definition describes the relation between input and output values of a Simulink subsystem at each step, and the method in this paper encodes in basically the same way.
Additionally, translation methods for multi-rate (or multi-periodic) models and conditionally executed subsystems have examined~\cite{Tripakis2005,Bourbouh2020}.
Our method handles multi-rates in the same way but differs in that we directly encode to SMT-LIB descriptions, whereas the above methods use the \verb|when| construct of Lustre.
Zhou et al.~\cite{Zhou2012} have proposed a translation method into input-output finite automata.
They also formalized conditioning on subsystems and multi-rate models.
The result of the transformation is a flat automaton with no subsystem structure.
Bouissou et al.~\cite{Bouissou2012a} have formalized the simulator for continuous-time models implemented in Simulink, which involves numerical integration and zero crossing detection.


There are SMT-based test generation methods that translate Simulink models to constraints and perform a symbolic analysis~\cite{Ren2016,Chakrabarti2016},
which is performed to obtain a test case as a solution that satisfies constraints.
They do not consider to verify invariance properties explicitly.
%
The SmartTestGen tool~\cite{Peranandam2012a,Raviram2012} combines four approaches of test generation;
one of them considers invariance checking to guide the coverage strategy.
From their evaluation, the effectiveness of the tool for our example models is not clear.

\section{Conclusions}

We have presented an SMT-based model checking method for the invariance verification of Simulink models.
Experimental result shows that it is useful in industrial setting;
we had the competing results with the state-of-the-art tool, SLDV;
our tool handled models that could not be properly analyzed by other tools.
The verification process is comprehensible by the intermediate encoded representation and the parameters such as $k$ and $d$.
The resulting invariance proofs are reliable based on exact encoding with bit vectors.

There are several future issues such as 
improvement of the model checking algorithm and the SMT-LIB encoding method for faster verification, and
experiments on the analysis of larger and more complex models.


%

\bibliographystyle{splncs04}
\bibliography{dld}

%
%

\end{document}